\documentclass[submission,nomarks,colorlinks=true]{dmtcs-episciences}

\usepackage[utf8]{inputenc}
\usepackage{subfigure}

\usepackage{amsmath,amssymb,amsthm}
\usepackage{algorithm}
\usepackage[noend]{algpseudocode}
\usepackage{enumerate}
\usepackage{comment}
\usepackage{tikz}
\usepackage{gensymb}
\usepackage{mathtools}

\newtheorem{example}{Example}
\newtheorem{definition}{Definition}

\newtheorem{theorem}{Theorem}
\newtheorem{lemma}{Lemma}

\newtheorem{corollary}{Corollary}

\newcommand{\cp}{\mathrel{\!\vbox{\offinterlineskip\ialign{%
    \hfil##\hfil\cr
    $\scriptstyle\circ$\cr
    \noalign{\kern-1.2ex}
    $\times$\cr
}}\!}}
\newcommand{\tcl}[1]{{#1^*}}
\newcommand{\ucl}[1]{{\mathcal{U}\!\left({#1}\right)}}
\newcommand{\ucmp}[1]{\overline{#1}}
\newcommand{\ucc}[1]{\ucmp{\tcl{#1}}}
\newcommand{\utcl}[1]{\ucl{\tcl{#1}}}
\newcommand{\forced}{\mathop{~\Gamma~}}
\newcommand{\fcl}{\mathop{~\Gamma^*~}}

\tikzstyle{vertex}=[circle,fill=black!10,minimum size=20pt,inner sep=0pt]
\tikzstyle{edge} = [draw,thick,->]
\tikzstyle{tedge} = [draw,dashed,->]
\tikzstyle{uedge} = [draw,thick]
\tikzset{top/.style={baseline=(current bounding box.north)}}
\tikzset{mid/.style={baseline=(current bounding box.center)}}
\tikzset{gscale/.style={xscale=1.4,yscale=1.5}}
\newcommand{\tikzcaption}[1]{\node[below=6mm] at (current bounding box.base) {#1}}

\newenvironment{calgorithm}[2]{%
\begin{center}
\begin{minipage}{\linewidth}
\vspace*{-.5em}
\begin{algorithm}[H]
\caption{~#1}\label{#2}
\begin{algorithmic}[1]
}{%
\end{algorithmic}
\end{algorithm}
\vspace*{-.5em}
\end{minipage}
\end{center}}

\author{Henning Koehler}
\title{A characterization of maximal 2-dimensional subgraphs of transitive graphs}
\affiliation{Massey University, New Zealand}
\received{2019-04-07}
\revised{-}
\accepted{-}

\begin{document}
\publicationdetails{VOL}{2019}{ISS}{NUM}{SUBM}
\maketitle

\begin{abstract}
~\\
A transitive graph is 2-dimensional if it can be represented as the intersection of two linear orders.
Such representations make answering of reachability queries trivial, and allow many problems that are NP-hard on arbitrary graphs to be solved in polynomial time.
One may therefore be interested in finding 2-dimensional graphs that closely approximate a given graph of arbitrary order dimension.

In this paper we show that the maximal 2-dimensional subgraphs of a transitive graph $G$ are induced by the optimal near-transitive orientations of the complement of $G$.
The same characterization holds for the maximal permutation subgraphs of a transitively orientable graph.
We provide an algorithm that enables this problem reduction in near-linear time, and an approach for enlarging non-maximal 2-dimensional subgraphs, such as trees.
\end{abstract}

\section{Introduction}

Approximating graphs that do not possess a simple structure with graphs that do has been a successful approach for many applications.
So far, trees (or forests) have been the primary tool of choice for this purpose -- e.g. approximating graphs with maximal sub-trees has led to the tree-cover indexing scheme~\cite{DBLP:conf/sigmod/AgrawalBJ89} for answering reachability queries, while the notion of tree-width has led to a wide range of tracktable algorithms for NP-hard problems.

Expanding the class of trees to the class of all 2-dimensional graphs should enable us to find closer approximations.
This does not make answering of reachability queries any harder, and many NP-hard problems that become easy on trees
can still be solved in polynomial time on arbitrary 2-dimensional graphs, including maximal clique, independent set, vertex cover, vertex coloring and clique cover~\cite{DBLP:journals/dm/McConnellS99}.

In the following, we aim to take a small step towards making 2-dimensional graphs a viable tool for approximation.
We will show that for a transitive graph $G$, any orientation of its complement induces a 2-dimensional subgraph of $G$, and that every maximal 2-dimensional subgraph can be induced in this way.
This reduces the problem of finding a maximal 2-dimensional subgraph to that of finding an orientation of its complement with minimal transitive closure.

To make this reduction effective, we provide a near-linear time algorithms for computing the 2-dimen\-sional subgraph induced by an orientation of the complement.
The reverse problem of finding a complement orientation inducing a maximal 2-dimensional subgraph can be solved with existing algorithms.

\section{Background}

To begin, we introduce some helpful tools and terminology.
Throughout the paper we consider all graphs to be simple and directed, and represent undirected edges as pairs of directed edges.

\begin{definition}[inverse, undirected, complement, oriented, transitive]~\\[-2em]
\begin{enumerate}[(i)]
\setlength\itemsep{-2pt}
\item The \emph{inverse} of an edge $(a,b)$ is the edge $(a,b)^{-1}=(b,a)$.
The inverse of a graph $G=(V,E)$ is the graph $G^{-1}=(V,E^{-1})$, where $E^{-1}=\{(a,b)^{-1}\mid (a,b)\in E\}$.
\item We say that a graph $G=(V,E)$ is \emph{undirected} iff $E=E^{-1}$.
The \emph{undirected closure} of a graph $G=(V,E)$ is the graph $\mathcal{U}(G)=(V,E^U)$, where $E^U=E\cup E^{-1}$.
\item The \emph{complement} of $G=(V,E)$ is the graph $\ucmp{G}=(V,V\cp V\setminus E^U)$, where $V\cp V=\{(a,b)\mid a,b\in V, a\neq b\}$.
\item A graph $G=(V,E)$ is \emph{oriented} iff $E\cap E^{-1}=\emptyset$.
A graph $G'=(V,E')$ is an \emph{orientation} of $G=(V,E)$ iff $E'$ is a maximal oriented subset of $E$.
\item A graph $G=(V,E)$ is \emph{transitive} iff for all edges $(a,b),(b,c)\in E$ with $a\neq c$ we also have $(a,c)\in E$.
The \emph{transitive closure} of a graph $G=(V,E)$ is $\tcl{G}=(V,\tcl{E})$, where $\tcl{E}$ is the minimal transitive superset of $E$.
\end{enumerate}
\end{definition}

Any acyclic transitive graph can be viewed as a partial order, and any partial order can be characterized as intersection of linear orders \cite{pos41}.
The minimal number of linear orders required for this is called the \emph{order dimension} of the graph.
In this paper, any references to dimension will be w.r.t. order dimension.
We call a graph \emph{$n$-dimensional} if its order dimension is $n$ or less.

In this work we shall focus exclusively on 2-dimensional graphs.
While the use of $n$-dimensional graph classes for larger $n$ would enable ever closer approximations, we note that there is no evidence that NP-hard problems tend to become easier when restricted to $n$-dimensional graphs for $n>2$.
In particular, deciding whether a graph is $n$-dimensional for $n>2$ has been shown to be NP-hard \cite{journals/siamadm/Yannakakis82}.

For most of our results we shall assume that given transitive graphs are acyclic.
For cyclic graphs the results can be applied to their condensation into strongly connected components.

\subsection{Permutation Graphs and Transitive Orientation}\label{S:TO}

We call a graph \emph{transitively orientable} iff it possesses a transitive orientation.
Transitive orientations of graphs are typically represented by providing a total ordering of vertices, which applies only to edges in the graph to be oriented.
As this representation only requires space (near)-linear in the number of vertices, it can be computed efficiently even if the graph is large -- e.g. the complement of a given sparse graph.

An undirected graph is a \emph{permutation graph} iff both it and its complement are transitively orientable.
It is well-known that any transitive orientation of a permutation graph is 2-dimensional, and that reversely the undirected closure of a 2-dimensional graph is a permutation graph.

\begin{lemma}[{\cite[Theorem 3.61]{pos41}}]\label{L:2d-toc}~\\
A transitive graph is 2-dimensional iff its complement is transitively orientable.
\end{lemma}

The linear orders describing a 2-dimensional digraph $G$ can be found by merging $G$ with the transitive orientation $\mathcal{O}$ of its complement, and with its inverse $\mathcal{O}^{-1}$.

\begin{example}\label{E:transitive-orientation}
Consider the graph shown on the left.
A transitive orientation of the complement of its transitive closure is given in the center.
Merging these graphs results in a linearization $L_1$.
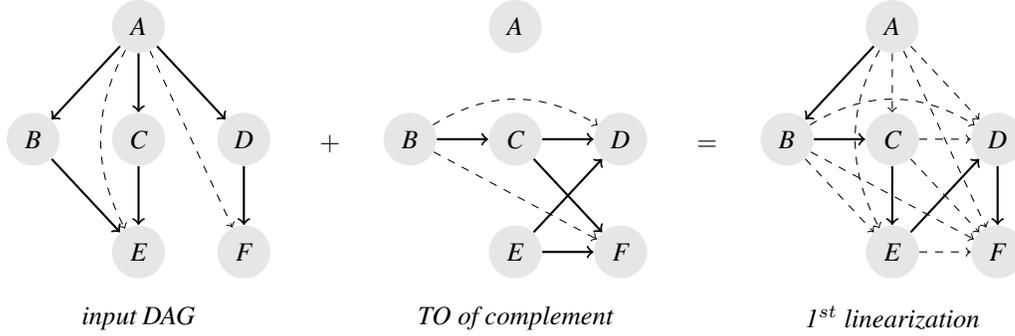
\newcommand{\vdef}{%
\foreach \pos/\name in {%
{(1,2)/A},
{(0,1)/B}, {(1,1)/C}, {(2,1)/D},
{(1,0)/E}, {(2,0)/F}}
    \node[vertex] (\name) at \pos {\name}}
\begin{center}
\hfill
\begin{tikzpicture}[gscale,top]
\vdef;
% edges
\foreach \source/\dest in {%
A/B, A/C, A/D, B/E, C/E, D/F}
    \path[edge] (\source) -- (\dest);
% transitive edges
\tikzstyle{edge} = [draw,dashed,->]
\path[edge] (A) to [out=-120,in=120] (E);
\path[edge] (A) -- (F);
\tikzcaption{input DAG};
\end{tikzpicture}
\hfill\raisebox{-2cm}{$+$}\hfill
\begin{tikzpicture}[gscale,top]
\vdef;
% edges
\foreach \source/\dest in {%
B/C, C/D, C/F, E/D, E/F}
    \path[edge] (\source) -- (\dest);
% transitive edges
\path[tedge] (B) -- (F);
\path[tedge] (B) to [bend left=30] (D);
\tikzcaption{TO of complement};
\end{tikzpicture}
\hfill\raisebox{-2cm}{$=$}\hfill
\begin{tikzpicture}[gscale,top]
\vdef;
% edges
\foreach \source/\dest in {%
A/B, B/C, C/E, E/D, D/F}
    \path[edge] (\source) -- (\dest);
% transitive edges
\foreach \source/\dest in {%
A/C, A/D, A/F, B/E, B/F, C/D, C/F, E/F}
    \path[tedge] (\source) -- (\dest);
\path[tedge] (A) to [bend right=30] (E);
\path[tedge] (B) to [bend left=30] (D);
\tikzcaption{1$^{st}$ linearization};
\end{tikzpicture}
\hfill~
\end{center}

Merging $G$ with the inverse orientation of its complement provides the second linearization $L_2$.
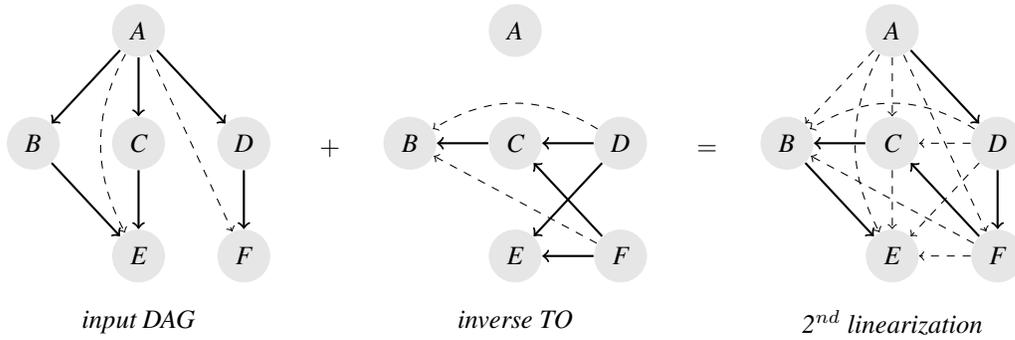
\begin{center}
\hfill
\begin{tikzpicture}[gscale,top]
\vdef;
% edges
\foreach \source/\dest in {%
A/B, A/C, A/D, B/E, C/E, D/F}
    \path[edge] (\source) -- (\dest);
% transitive edges
\tikzstyle{edge} = [draw,dashed,->]
\path[edge] (A) to [out=-120,in=120] (E);
\path[edge] (A) -- (F);
\tikzcaption{input DAG};
\end{tikzpicture}
\hfill\raisebox{-2cm}{$+$}\hfill
\begin{tikzpicture}[gscale,top]
\vdef;
% edges
\foreach \source/\dest in {%
C/B, D/C, F/C, D/E, F/E}
    \path[edge] (\source) -- (\dest);
% transitive edges
\path[tedge] (F) -- (B);
\path[tedge] (D) to [bend right=30] (B);
\tikzcaption{inverse TO};
\end{tikzpicture}
\hfill\raisebox{-2cm}{$=$}\hfill
\begin{tikzpicture}[gscale,top]
\vdef;
% edges
\foreach \source/\dest in {%
A/D, D/F, F/C, C/B, B/E}
    \path[edge] (\source) -- (\dest);
% transitive edges
\foreach \source/\dest in {%
A/C, A/B, A/F, D/C, D/E, F/B, F/E, C/E}
    \path[tedge] (\source) -- (\dest);
\path[tedge] (A) to [bend right=30] (E);
\path[tedge] (D) to [bend right=30] (B);
\tikzcaption{2$^{nd}$ linearization};
\end{tikzpicture}
\hfill~
\end{center}
The partial ordering induced by $G$ can now be represented as $L_1\cap L_2$.
\qed
\end{example}

An efficient algorithm for combining the partial order described by $G^*$ with a transitive orientation of its complement, given as a linear ordering, can e.g. be found in \cite[Proof of Lemma~2]{DBLP:journals/jacm/EvenPL72}.
It is reproduced as Algorithm~\ref{A:merge} below.
The \emph{rank} of a vertex is the number of smaller vertices w.r.t. to some linear ordering.

\begin{calgorithm}{Merge $G$ with transitive orientation of complement}{A:merge}
\Require{Acyclic digraph $G$, linear order $L_{\overline{G}}$ describing transitive orientation $\mathcal{O}$ of $\ucc{G}$}
\Ensure{Linear extension $L$ of $G$ and $\mathcal{O}$}
\Procedure{Merge}{$G,L_{\overline{G}}$}
\State $L=[~]$
\State $S=\{ v\in V \mid v\text{ has no incoming edges}\}$
\While{$S\neq\emptyset$}
\State remove from $S$ vertex $s\in S$ with minimal rank in $L_{\overline{G}}$
\State append $s$ to $L$
\State add successors of $s$ with no other incoming arcs to $S$
\State remove $s$ from $G$
\EndWhile
\State \textbf{return} $L$
\EndProcedure
\end{calgorithm}

Using a heap to represent $S$, we can update $S$ and retrieve the minimal rank vertex $s$ in $O(\log n)$, resulting in an overall running time of $O(m + n\cdot\log n)$.

The issue of finding a transitive orientation of a transitively orientable graph has received much attention in the literature.
We point to ordered vertex partitioning \cite{DBLP:journals/dmtcs/McConnellS00} as a practical algorithm that runs in near-linear time.
A transitive orientation of the complement of a graph, again described as a linear ordering of vertices, can also be obtained in near-linear time with a little tweak \cite[Section 11.2]{DBLP:journals/dmtcs/DahlhausGM02}.

\subsection{Forced Orientation}

At the core of transitive orientations is the forcing relationship between edges, where orientation of one edge forces the orientation of another to maintain transitivity \cite{gallai67}.

\begin{definition}[forced orientation]~\\
Let $G$ be an undirected graph.
We say that two edges $(a,b),(c,d)$ in $G$ \emph{force each other}, denoted as $(a,b)\forced(c,d)$, iff either $a=c$ and $b,d$ are not adjacent, or $b=d$ and $a,c$ are not adjacent.
We denote the transitive closure of relation $\forced$ by $\fcl$.
\end{definition}

The forcing relationship can be linked to the transitivity of orientations as follows:

\begin{lemma}[{\cite[Chapter 5]{Golumbic:2004:AGT:984029}}]\label{L:trans-forced}
An orientation $\mathcal{O}$ of $G$ is transitive iff\\[-1.5em]
\begin{enumerate}[(i)]
\setlength\itemsep{-2pt}
\item $\mathcal{O}$ contains all edges forced by edges in $\mathcal{O}$, and
\item $\mathcal{O}$ does not contain any cycles.
\end{enumerate}
\end{lemma}

\begin{example}\label{E:forced}
Consider the graph $G$ below, which is the complement of the graph from Example~\ref{E:transitive-orientation}.
If we pick an orientation of one edge, e.g. $E\to F$, then the orientation of most other edges is forced by it.
E.g. $E\to F$ forces $E\to D$, which in turn forces $B\to D$.\\[-2em]
\newcommand{\vdef}{%
\foreach \pos/\name in {%
{(1,2)/A},
{(0,1)/B}, {(1,1)/C}, {(2,1)/D},
{(1,0)/E}, {(2,0)/F}}
    \node[vertex] (\name) at \pos {\name}}
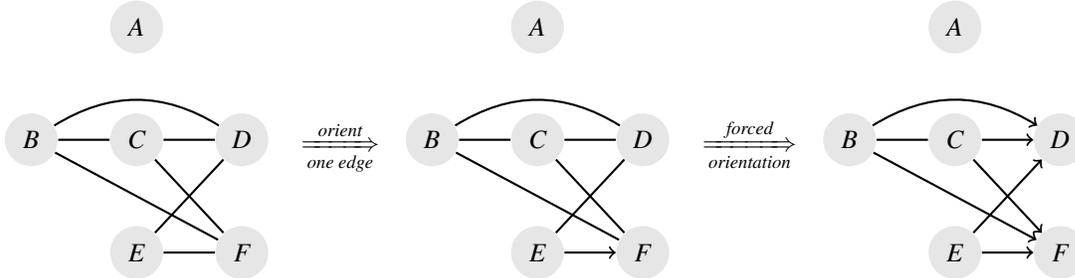
\begin{center}
\begin{tikzpicture}[gscale,top]
\vdef;
% edges
\foreach \source/\dest in {%
B/C, B/F, C/D, C/F, E/D, E/F}
    \path[uedge] (\source) -- (\dest);
\path[uedge] (B) to [bend left=30] (D);
\end{tikzpicture}
\quad\raisebox{-2cm}{$\xRightarrow[\text{one edge}]{\text{orient}}$}\quad
\begin{tikzpicture}[gscale,top]
\vdef;
% edges
\foreach \source/\dest in {%
B/C, B/F, C/D, C/F, E/D}
    \path[uedge] (\source) -- (\dest);
\path[uedge] (B) to [bend left=30] (D);
% oriented edges
\path[edge] (E) to (F);
\end{tikzpicture}
\quad\raisebox{-2cm}{$\xRightarrow[\text{orientation}]{\text{forced}}$}\quad
\begin{tikzpicture}[gscale,top]
\vdef;
% edges
\path[uedge] (B) to (C);
% oriented edges
\foreach \source/\dest in {%
B/F, C/D, C/F, E/D, E/F}
    \path[edge] (\source) -- (\dest);
\path[edge] (B) to [bend left=30] (D);
\end{tikzpicture}
\end{center}
Orienting the remaining unforced edge as $B\to C$ yields the transitive orientation from Example~\ref{E:transitive-orientation}.
Orienting it as $C\to B$ yields another transitive orientation.
\qed
\end{example}

\section{Characterizing 2-dimensional Subgraphs}\label{S:reduction}

We shall now relate 2-dimensional subgraphs of a transitive graph $G$ with orientations of its complement.
It turns out that any maximal 2-dimensional subgraph of $G$ can be obtained by removing edges that lie in the undirected closure of the transitive closure of some orientation of $\ucmp{G}$.
Here we use the term \emph{near-transitive} orientation to indicate that we are interested in orientations whose transitive closure is small.

Consider now a directed acyclic graph $G$ which is \emph{not} transitive.
A critical observation in relating transitive orientations of $\ucc{G}$ and $\ucmp{G}$, if they exist, is that the transitive closures of their forcing relationships are closely related.
We shall say that two edges \emph{directly} force each other if they are related via $\forced\!$, and that they \emph{indirectly} force each other if they are related via $\fcl\!$.

\begin{lemma}\label{L:forcing}
Let $(a,b),(c,d)$ be two edges in $\ucc{G}$ that indirectly force each other in $\ucc{G}$, i.e., we have $(a,b) \fcl (c,d)$.
Then $(a,b),(c,d)$ indirectly force each other in $\ucmp{G}$.
\end{lemma}

\begin{proof}
Let $(x,a),(x,b)$ be two edges in $\ucc{G}$ that force each other directly.
Then either $(a,b)$ or $(b,a)$ lies in $\tcl{G}$, say $(a,b)$.
This means $G$ contains a path $a\to v_1 \to \ldots \to v_n \to b$, for some $n$.

Assume $(x,v_i)$ does not lie in $\ucmp{G}$, for $1\leq i\leq n$.
Then either $(x,v_i)$ or $(v_i,x)$ must lie in $G$.
If $(x,v_i)$ lies in $G$, then $(x,b)$ lies in $\tcl{G}$, contradicting $(x,b)\in\ucc{G}$.
If $(v_i,x)$ lies in $G$, then $(a,x)$ lies in $\tcl{G}$, contradicting $(x,a)\in\ucc{G}$.

Hence $(x,v_i)$ must lie in $\ucmp{G}$ for all $i=1\ldots n$, so the following forcing relationships hold in $\ucmp{G}$:
\[
(x,a)
\quad\Gamma\quad (x,v_1)
\quad\Gamma\quad \ldots
\quad\Gamma\quad (x,v_n)
\quad\Gamma\quad (x,b)
\]

This shows the claim for edges that directly force each other in $\ucc{G}$.
For edges that force each other indirectly, the claim then follows by transitivity of $\fcl$.
\end{proof}

As a consequence, any transitive orientation of $\ucmp{G}$ is a transitive orientation of $\ucc{G}$ as well:

\begin{theorem}\label{T:toc-tocc}
Let $H$ be a transitive orientation of $\ucmp{G}$.
Then $H\cap\ucc{G}$ is a transitive orientation of $\ucc{G}$.
\end{theorem}

\begin{proof}
We show that conditions (i) and (ii) of Lemma~\ref{L:trans-forced} hold for $H\cap\ucc{G}$.
\begin{enumerate}[(i)]
\item By Lemma~\ref{L:forcing} any two edges in $\ucc{G}$ forcing each other also force each other in $\ucmp{G}$.
Thus condition (i) for $H$ implies condition (i) for $H\cap\ucc{G}$.
\item Since $H$ is cycle-free, any subgraph of $H$ must be as well.
\end{enumerate}
\vspace{-2em}
\end{proof}

In particular, for any transitive graph $G$, its complement $\ucmp{G}$ is a graph with transitively orientable complement.
This allows us to construct a 2-dimensional subgraph of $G$:

\begin{corollary}\label{Cor:tocc-2d}
Let $G$ be transitive and $H$ an orientation of $\ucmp{G}$.
Then $G\setminus\utcl{H}$ is 2-dimensional.
\end{corollary}

\begin{proof}
As $G$ is a transitive orientation of $\ucmp{H}$, it follows from Theorem~\ref{T:toc-tocc} that $G$ induces a transitive orientation on $\ucc{H}$, namely $G\setminus\utcl{H}$.
Thus $G\setminus\utcl{H}$ is a transitive graph, and its complement has a transitive orientation $\tcl{H}$, which makes it 2-dimensional by Lemma~\ref{L:2d-toc}.
\end{proof}

We illustrate this construction with an example.

\begin{example}\label{E:induced}
Consider the graph $G$ shown on the left.
An orientation $H$ of $\ucmp{G}$ is shown in the center.
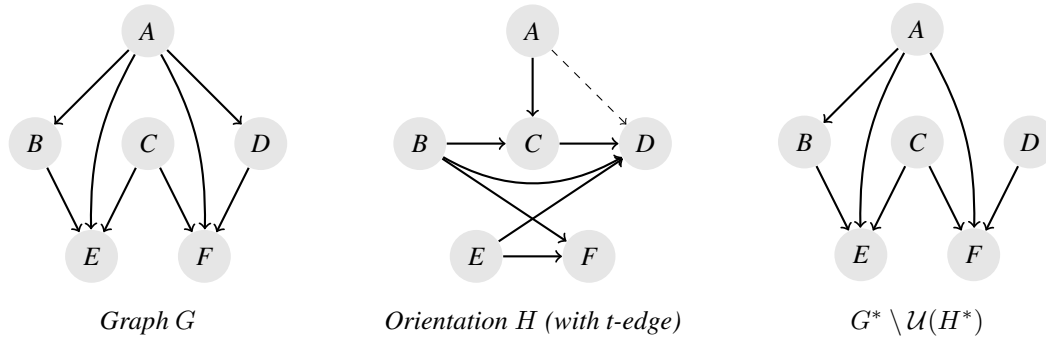
\begin{center}
\begin{tikzpicture}[scale=1.5,mid]
% vertices
\foreach \pos/\name in {%
{(1,2)/A},
{(0,1)/B}, {(1,1)/C}, {(2,1)/D},
{(.5,0)/E}, {(1.5,0)/F}}
    \node[vertex] (\name) at \pos {\name};
% edges
\foreach \source/\dest in {%
A/B, A/D, B/E, C/E, C/F, D/F}
    \path[edge] (\source) -- (\dest);
% transitive edges
\path[edge] (A) to [bend right=15] (E);
\path[edge] (A) to [bend left =15] (F);
\tikzcaption{Graph $G$};
\end{tikzpicture}
\hspace{1cm}
\begin{tikzpicture}[scale=1.5,mid]
% vertices
\foreach \pos/\name in {%
{(1,2)/A},
{(0,1)/B}, {(1,1)/C}, {(2,1)/D},
{(.5,0)/E}, {(1.5,0)/F}}
    \node[vertex] (\name) at \pos {\name};
% edges
\foreach \source/\dest in {%
A/C, B/C, B/F, C/D, E/D, E/F}
    \path[edge] (\source) -- (\dest);
\path[edge] (B) to [bend right=30] (D);
% transitive edges
\path[draw,dashed,->] (A) -- (D);
\tikzcaption{Orientation $H$ (with t-edge)};
\end{tikzpicture}
\hspace{1cm}
\begin{tikzpicture}[scale=1.5,mid]
% vertices
\foreach \pos/\name in {%
{(1,2)/A},
{(0,1)/B}, {(1,1)/C}, {(2,1)/D},
{(.5,0)/E}, {(1.5,0)/F}}
    \node[vertex] (\name) at \pos {\name};
% edges
\foreach \source/\dest in {%
A/B, B/E, C/E, C/F, D/F}
    \path[edge] (\source) -- (\dest);
% transitive edges
\path[edge] (A) to [bend right=15] (E);
\path[edge] (A) to [bend left =15] (F);
\tikzcaption{$\tcl{G}\setminus\ucl{\tcl{H}}$};
\end{tikzpicture}
\end{center}
Observe that the only transitivity violation occurs for $A\to C,C\to D$, causing $\tcl{H}$ to contain the extra edge $A\to D$.
When we remove this edge from $\tcl{G}$, we obtain a 2-dimensional subgraph.\qed
\end{example}

We note that not every 2-dimensional subgraph can be constructed using Corollary~\ref{Cor:tocc-2d} (consider e.g. a digraph representing a linear ordering).
However, when approximating $G$ with a 2-dimensional subgraph, we want this subgraph to be as big as possible, either w.r.t. the subgraph relationship (local optimality), or w.r.t. the number of edges (global optimality).
Thus it would be sufficient if all \emph{maximal} 2-dimensional subgraphs could be constructed in this fashion -- and as it turns out, this is indeed the case.

In the following we shall refer to $G\setminus\utcl{H}$ as the subgraph of $G$ \emph{induced} by $H$.

\begin{lemma}\label{L:improve-approx}
Let $G$ be transitive and $S$ a 2-dimensional subgraph of $G$.
Then $\ucmp{S}$ is transitively orientable, and for any transitive orientation $H_S$ of $\ucmp{S}$
the subgraph $S':=G\setminus\utcl{H}$ induced by $H:=H_S\cap\ucmp{G}$ is a 2-dimensional supergraph of $S$.
\end{lemma}

\begin{proof}
\setlength\abovedisplayskip{5pt}
\setlength\belowdisplayskip{-1em}
$\ucmp{S}$ is transitively orientable by Lemma~\ref{L:2d-toc},
and by Corollary~\ref{Cor:tocc-2d} $S'$ is 2-dimensional, so it remains to show $S\subseteq S'$.
From $H\subseteq H_S$ and transitivity of $H_S$ we get $\tcl{H}\subseteq H_S$, and thus
\[
S' = G\setminus\utcl{H} \supseteq G\setminus\ucl{H_S} = G\setminus\ucmp{S} = S
\]
\end{proof}

Together with Corollary~\ref{Cor:tocc-2d}, this give us the following:

\begin{theorem}\label{T:induced}
Let $G$ be transitive. Then\\[-1.5em]
\begin{enumerate}[(i)]
\setlength\itemsep{-2pt}
\item every orientation $H$ of $\ucmp{G}$ induces a 2-dimensional subgraph of $G$, and
\item every locally maximal 2-dimensional subgraph of $G$ is induced by some orientation $H$ of $\ucmp{G}$.
\end{enumerate}
\end{theorem}

\begin{proof}
Condition (i) is just a restatement of Corollary~\ref{Cor:tocc-2d}.

To show (ii) let $S$ be a locally maximal 2-dimensional subgraph of $G$.
By Lemma~\ref{L:improve-approx} there exists a 2-dimensional subgraph $S'$ with $S\subseteq S'$ such that $S'$ is induced by some orientation $H$ of $\ucmp{G}$.
But since $S$ is locally maximal we must have $S'=S$, so $S$ is induced by $H$.
\end{proof}

Theorem~\ref{T:induced} reduces the problem of finding a maximal 2-dimensional subgraph of $G$ to that of finding an optimal near-transitive orientation of $\ucmp{G}$, that is, an orientation with minimal transitive closure.

\subsection{Undirected graphs}

As seen in Lemma~\ref{L:2d-toc}, the notions of 2-dimensional graph and permutation graph are closely related.
This relationship can be strengthened further for maximal subgraphs.

\begin{lemma}\label{L:2d-perm}
Let $G$ be transitive and $S$ a subgraph of $G$.
Then $S$ is a (locally) maximal 2-dimensional subgraph of $G$ iff $\ucl{S}$ is a (locally) maximal permutation subgraph of $\ucl{G}$.
\end{lemma}

\begin{proof}
If $\ucl{S}$ is a permutation graph, then $\ucmp{S}$ is transitively orientable, and by Theorem~\ref{T:toc-tocc} so is $\ucc{S}$.
This makes $\ucl{\tcl{S}}$ a permutation subgraph of $\ucl{G}$.
If $\ucl{S}$ is maximal, then we must have $\ucl{\tcl{S}}=\ucl{S}$ and thus $\tcl{S}=S$, so $S$ is 2-dimensional by Lemma~\ref{L:2d-toc}.

Conversely, if $S$ is a 2-dimensional subgraph of $G$, it follows by Lemma~\ref{L:2d-toc} that $\ucl{S}$ is a permutation subgraph of $\ucl{G}$.
Maximality of one now clearly implies maximality of the other.
\end{proof}

As a consequence, permutation subgraphs of a transitively orientable graph can be characterized in the same way as 2-dimensional subgraphs of a transitive graph.

\begin{corollary}\label{C:induced-undirected}
Let $G$ be transitively orientable. Then\\[-1.5em]
\begin{enumerate}[(i)]
\setlength\itemsep{-2pt}
\item every orientation $H$ of $\ucmp{G}$ induces a permutation subgraph of $G$, and
\item every locally maximal permutation subgraph of $G$ is induced by some orientation $H$ of $\ucmp{G}$.
\end{enumerate}
\end{corollary}

\subsection{Improving Approximations}\label{S:improving}

Ideally we would like an efficient algorithm for finding optimal near-transitive orientations.
Until such an algorithm is discovered (if one exists -- the hardness of this optimization problem is currently open), 
an alternative could be to use Lemma~\ref{L:improve-approx} to improve a given 2-dimensional subgraph, such as the one provided by the tree-cover algorithm \cite{DBLP:conf/sigmod/AgrawalBJ89}.
The example below illustrates this approach.

\begin{example}
Consider the following graph $G$, and a tree-cover $T$ of $G$ (transitive edges omitted):\\
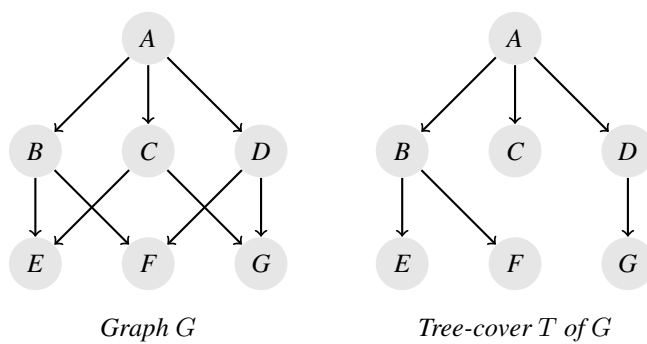
\begin{center}
\begin{tikzpicture}[scale=1.5,mid]
% First we draw the vertices
\foreach \pos/\name in {%
{(1,2)/A},
{(0,1)/B}, {(1,1)/C}, {(2,1)/D},
{(0,0)/E}, {(1,0)/F}, {(2,0)/G}}
    \node[vertex] (\name) at \pos {\name};
% Connect vertices with edges
\foreach \source/\dest in {%
A/B, A/C, A/D, B/E, B/F, C/E, C/G, D/F, D/G}
    \path[edge] (\source) -- (\dest);
\tikzcaption{Graph $G$};
\end{tikzpicture}
\hspace{1cm}
\begin{tikzpicture}[scale=1.5,mid]
% First we draw the vertices
\foreach \pos/\name in {%
{(1,2)/A},
{(0,1)/B}, {(1,1)/C}, {(2,1)/D},
{(0,0)/E}, {(1,0)/F}, {(2,0)/G}}
    \node[vertex] (\name) at \pos {\name};
% Connect vertices with edges
\foreach \source/\dest in {%
A/B, A/C, A/D, B/E, B/F, D/G}
    \path[edge] (\source) -- (\dest);
\tikzcaption{Tree-cover $T$ of $G$};
\end{tikzpicture}
\end{center}
As $T$ is 2-dimensional, we can obtain a transitive orientation $H_T$ of its complement.\\
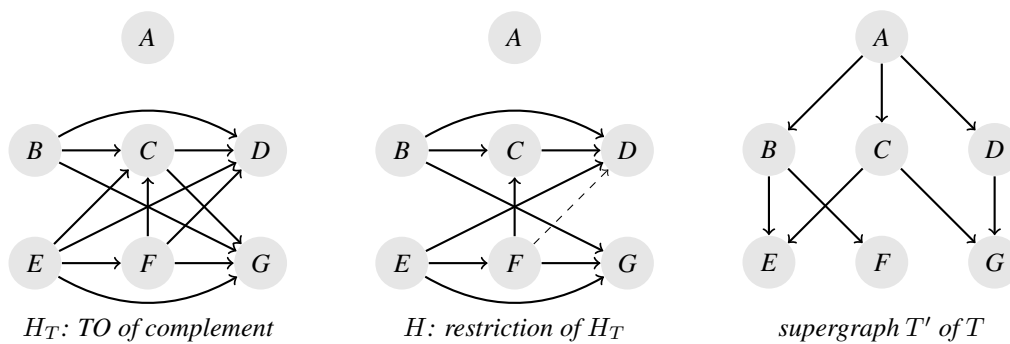
\begin{center}
\begin{tikzpicture}[scale=1.5,mid]
% First we draw the vertices
\foreach \pos/\name in {%
{(1,2)/A},
{(0,1)/B}, {(1,1)/C}, {(2,1)/D},
{(0,0)/E}, {(1,0)/F}, {(2,0)/G}}
    \node[vertex] (\name) at \pos {\name};
% Connect vertices with edges
\foreach \source/\dest in {%
B/C, B/G, C/D, C/G, E/C, E/D, E/F, F/C, F/D, F/G}
    \path[edge] (\source) -- (\dest);
\path[edge] (B) to [bend left=30] (D);
\path[edge] (E) to [bend right=30] (G);
\tikzcaption{$H_T$: TO of complement};
\end{tikzpicture}
\hspace{1cm}
\begin{tikzpicture}[scale=1.5,mid]
% First we draw the vertices
\foreach \pos/\name in {%
{(1,2)/A},
{(0,1)/B}, {(1,1)/C}, {(2,1)/D},
{(0,0)/E}, {(1,0)/F}, {(2,0)/G}}
    \node[vertex] (\name) at \pos {\name};
% Connect vertices with edges
\foreach \source/\dest in {%
B/C, B/G, C/D, E/D, E/F, F/C, F/G}
    \path[edge] (\source) -- (\dest);
\path[edge] (B) to [bend left=30] (D);
\path[edge] (E) to [bend right=30] (G);
% transitive edges
\path[draw,dashed,->] (F) -- (D);
\tikzcaption{$H$: restriction of $H_T$};
\end{tikzpicture}
\hspace{1cm}
\begin{tikzpicture}[scale=1.5,mid]
% First we draw the vertices
\foreach \pos/\name in {%
{(1,2)/A},
{(0,1)/B}, {(1,1)/C}, {(2,1)/D},
{(0,0)/E}, {(1,0)/F}, {(2,0)/G}}
    \node[vertex] (\name) at \pos {\name};
% Connect vertices with edges
\foreach \source/\dest in {%
A/B, A/C, A/D, B/E, B/F, C/E, C/G, D/G}
    \path[edge] (\source) -- (\dest);
\tikzcaption{supergraph $T'$ of $T$};
\end{tikzpicture}
\end{center}
Restricting $H_T$ to $\ucmp{G}$ give us the near-transitive graph $H$ shown in the center.
Removing from $G$ the edges in $\ucl{\tcl{H}}$, namely $D\to F$, results in the 2-dimensional supergraph $T'$ of $T$.\qed
\end{example}

For some 2-dimensional subgraphs $S$, a transitive orientation of $\ucmp{S}$ can be found easily, e.g. when $S$ is a tree.
More general cases could be handled by any algorithm for finding transitive orientations -- in particular, the ordered vertex partitioning approach of \cite{DBLP:journals/dmtcs/McConnellS00}, together with a tweak from \cite[Section 11.2]{DBLP:journals/dmtcs/DahlhausGM02}, allows us to construct a transitive orientation of $\ucmp{S}$ in time near-linear in the size of $S$.

\section{Computing Induced Subgraphs}

In order to make the problem reduction described in Section~\ref{S:reduction} effective, we require a fast algorithm for constructing the 2-dimensional subgraph induced by an orientation of the complement.

Since $G_H:=G\setminus\utcl{H}$ is 2-dimensional, we can describe it with two linear orderings, such as
$G_H\cup\tcl{H}$ and $G_H\cup(\tcl{H})^{-1}$, as illustrated in Example~\ref{E:transitive-orientation}.
We therefore require an algorithm which takes $G$ and a linear ordering describing $H$ as input, and returns the linear ordering describing $G_H\cup\tcl{H}$.
As $H$ can be much larger than $G$, we want this algorithm to be near-linear in the size of $G$.

\begin{example}
Consider again the graph from Example~\ref{E:induced}, reproduced below.
The linear ordering BEACDF, short for $B<E<A<C<D<F$, describes the orientation of its complement given in the center.
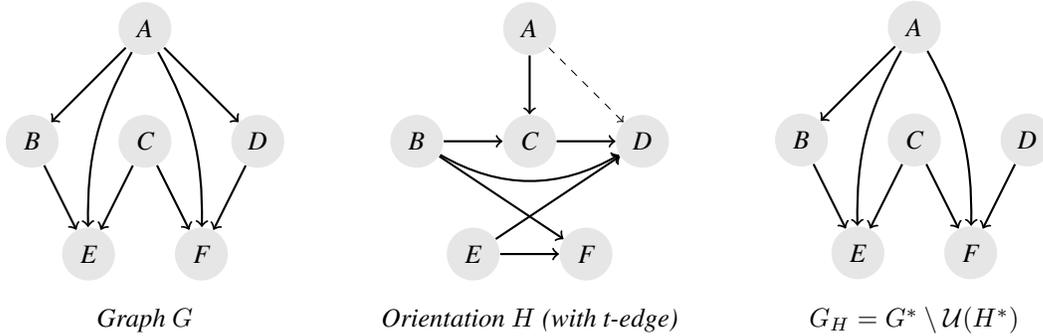
\begin{center}
\begin{tikzpicture}[scale=1.5,mid]
% vertices
\foreach \pos/\name in {%
{(1,2)/A},
{(0,1)/B}, {(1,1)/C}, {(2,1)/D},
{(.5,0)/E}, {(1.5,0)/F}}
    \node[vertex] (\name) at \pos {\name};
% edges
\foreach \source/\dest in {%
A/B, A/D, B/E, C/E, C/F, D/F}
    \path[edge] (\source) -- (\dest);
% transitive edges
\path[edge] (A) to [bend right=15] (E);
\path[edge] (A) to [bend left =15] (F);
\tikzcaption{Graph $G$};
\end{tikzpicture}
\hspace{1cm}
\begin{tikzpicture}[scale=1.5,mid]
% vertices
\foreach \pos/\name in {%
{(1,2)/A},
{(0,1)/B}, {(1,1)/C}, {(2,1)/D},
{(.5,0)/E}, {(1.5,0)/F}}
    \node[vertex] (\name) at \pos {\name};
% edges
\foreach \source/\dest in {%
A/C, B/C, B/F, C/D, E/D, E/F}
    \path[edge] (\source) -- (\dest);
\path[edge] (B) to [bend right=30] (D);
% transitive edges
\path[draw,dashed,->] (A) -- (D);
\tikzcaption{Orientation $H$ (with t-edge)};
\end{tikzpicture}
\hspace{1cm}
\begin{tikzpicture}[scale=1.5,mid]
% vertices
\foreach \pos/\name in {%
{(1,2)/A},
{(0,1)/B}, {(1,1)/C}, {(2,1)/D},
{(.5,0)/E}, {(1.5,0)/F}}
    \node[vertex] (\name) at \pos {\name};
% edges
\foreach \source/\dest in {%
A/B, B/E, C/E, C/F, D/F}
    \path[edge] (\source) -- (\dest);
% transitive edges
\path[edge] (A) to [bend right=15] (E);
\path[edge] (A) to [bend left =15] (F);
\tikzcaption{$G_H=\tcl{G}\setminus\ucl{\tcl{H}}$};
\end{tikzpicture}
\end{center}
We are looking for an algorithm that combines $G$ and BEACDF into ABCEDF describing $G_H\cup\tcl{H}$.
The same algorithm can then be employed to combine $G$ and the inverse orientation FDCAEB into DCAFBE describing $G_H\cup(\tcl{H})^{-1}$.
$G_H$ is now represented as the intersection of ABCEDF and DCAFBE.\qed
\end{example}

To design such an algorithm, we employ the same approach as in Algorithm~\ref{A:merge}, that is:
\begin{enumerate}
\setlength\itemsep{-2pt}
\item We keep track of all vertices whose ancestors in $H$ have been processed, and
\item among those we pick the minimal one w.r.t. the ordering imposed by $G$.
\end{enumerate}

As nodes will have a large number of ancestors in $H$ when $G$ is sparse, we do not track the set of unprocessed ancestors, but only their number, which we shall refer to as the \emph{countdown} of a vertex.
The initial countdown value for a vertex $v$ can easily be computed as
\newcommand\hrank{\text{rank}_H}
\[
\text{countdown}(v) = \hrank(v) - |\{\; (v,w)\in \ucl{G} \mid \hrank(w) < \hrank(v) \;\}|
\]
where $\hrank(v)$ denotes the number of vertices smaller than $v$ in the linear ordering describing $H$.

When we process a vertex $v$, i.e., append it to the current output ordering, we would then need to reduce the countdown for all descendants of $v$ in $H$.
To avoid the complexity of doing this, we instead \emph{increase} the countdown for all \emph{non-}descendants of $v$ in $H$, which modifies the countdown function to denote the number of unprocessed ancestors plus the number of processed vertices.
Thus all ancestors of $v$ have been processed once its countdown reaches the number of processed vertices, which is easy to track.

Note though that the number of non-descendants of $v$ in $H$ can still be huge (on average it will be even larger than the number of descendants).
However, non-descendants of $v$ are either ancestors of $v$ in $H$ or neighbors of $v$ in $G$.
As ancestors of $v$ must have been processed already by the time we process $v$, we no longer need to keep their countdown updated.
This leaves only (unprocessed) neighbors of $v$ in $G$ to be updated, and the total number of those is bounded by the size of $G$.

We summarize this description as Algorithm~\ref{A:cmerge}.

\newcommand\ctd{\text{countdown}}
\begin{calgorithm}{Merge $H$ with transitive orientation of complement}{A:cmerge}
\Require{Transitive DAG $G$, linear order $L_H$ describing orientation $H$ of $\ucmp{G}$}
\Ensure{Linear order $L$ describing $G_H\cup\tcl{H}$}
\Procedure{Complement-Merge}{$G,L_H$}
\For{$v\in V$}
\State $\ctd(v) \gets \hrank(v) - |\{(v,w)\in \ucl{G} \mid \hrank(w) < \hrank(v) \}|$
%\State $\ctd(v) \gets \hrank(v)$
\EndFor
%\For{$(v,w)\in \ucl{G}$ with $\hrank(w) < \hrank(v)$}
%\State $\ctd(v) \gets \ctd(v) - 1$
%\EndFor
\State $S\gets\{ \}$; $L\gets[~]$; $L_G \gets$ some linearization of $G$
\For{$i=0..(|V|-1)$}
\State add vertices $v$ with $\ctd(v)=i$ to $S$
\State remove from $s$ vertex $s\in S$ with minimal rank in $L_G$
\State append $s$ to $L$
\For{neighbors $v$ of $s$ in $G$}
\If{$\ctd(v) > i$}
\State $\ctd(v) \gets \ctd(v) + 1$
\EndIf
\EndFor
\EndFor
\State \textbf{return} $L$
\EndProcedure
\end{calgorithm}

Correctness of Algorithm~\ref{A:cmerge} should be evident from the preceeding discussion.
The first for-loop can be implemented in $O(m)$ by iterating once over the edges in $\ucl{G}$.
Using a suitable data structure for tracking vertices by countdown,
e.g. an array of lists of vertices together with pointers into these lists,
and a heap for $S$ as in Algorithm~\ref{A:merge},
one can implement the second for-loop to run in $O(m + n\cdot\log n)$.

\section{Conclusion}

With Theorem~\ref{T:induced} we have established a tight relationship between 2-dimensional subgraphs of a transitive graph $G$ (or permutation subgraphs of a transitively orientable graph $G$), and orientations of the complement of $G$.
Together with Algorithm~\ref{A:cmerge} for computing the 2-dimensional subgraph induced by such an orientation,
this reduces the problem of finding maximal 2-dimensional subgraphs (or permutation subgraphs) to that of finding optimal near-transitive orientations.

This problem reduction is of course only a first step, with many of open problems that still need to be addressed.
Foremost is the issue of finding optimal near-transitive orientations - a polynomial time algorithm or proof of the problem's hardness would be the logical next step here.
Until then, the approach described in Section~\ref{S:improving} for extending an existing 2-dimensional subgraph may be of some use.

Furthermore, transitive graphs typically arise as the transitive closure of a given graph (e.g. for reachability queries).
It would thus be helpful if Algorithm~\ref{A:cmerge} could be improved to work with non-transitive graphs, without explicitly computing the transitive closure first.

Finally, we briefly suggested that 2-dimensional subgraphs might be a suitable replacement for trees in the design of tracktable algorithm.
This idea still remains to be explored.

\bibliography{graph.bib}{}
\bibliographystyle{plain}

\end{document}